\documentclass[12pt]{article}

\usepackage{amsfonts, amsmath, amsthm}
\usepackage{geometry}
\usepackage{hyperref, graphicx}
\usepackage{mathtools}
\usepackage[font=footnotesize]{caption}

\newcounter{dummy} \numberwithin{dummy}{section}
\newtheorem{theorem}[dummy]{Theorem}
\newtheorem{definition}[dummy]{Definition}

\newtheorem{lemma}[dummy]{Lemma}

\newtheorem{remark}[dummy]{Remark}
\numberwithin{figure}{section}

\def\({\left(}
\def\){\right)}
\newcommand{\wt}{\widetilde} 
\newcommand{\floor}[1]{\left\lfloor #1 \right\rfloor}
\newcommand{\ceil}[1]{\left\lceil #1 \right\rceil}
\newcommand{\bemph}[1]{{\bf #1}}

\def\R{\mathbb{R}}

\def\E{\mathbb{E}}
\def\P{\mathbb{P}}
\def\1{\mathbf{1}}

\def\D{\mathcal{D}}
\def\L{\mathcal{L}}

\DeclareMathOperator{\Binom}{Binom}
\let\Re\relax
\DeclareMathOperator{\Re}{Re}
\DeclareMathOperator{\ind}{Ind}
\DeclareMathOperator{\indset}{IndSet}

\def\bw{\mathbf{w}}
\def\bx{\mathbf{x}}
\def\by{\mathbf{y}}
\def\bz{\mathbf{z}}

\def\bt{\mathbf{t}}
\def\ba{\mathbf{a}}

\def\bX{\mathbf{X}}

\begin{document}
\title{Average-case reconstruction for the deletion channel: subpolynomially many traces suffice}
\author{
  Yuval Peres
  \thanks{Microsoft Research; \texttt{peres@microsoft.com}}
  \and
  Alex Zhai
  \thanks{Stanford University; \texttt{azhai@stanford.edu}}
}
\maketitle





\begin{abstract}
  The \emph{deletion channel} takes as input a bit string $\bx \in
  \{0,1\}^n$, and deletes each bit independently with probability $q$,
  yielding a shorter string. The \emph{trace reconstruction problem}
  is to recover an unknown string $\bx$ from many independent outputs
  (called ``traces'') of the deletion channel applied to $\bx$.

  We show that if $\bx$ is drawn uniformly at random and $q < 1/2$,
  then $e^{O(\log^{1/2} n)}$ traces suffice to reconstruct $\bx$ with
  high probability. The previous best bound, established in 2008 by
  Holenstein, Mitzenmacher, Panigrahy, and Wieder \cite{HMPW08}, uses
  $n^{O(1)}$ traces and only applies for $q$ less than a smaller
  threshold (it seems that $q < 0.07$ is needed).

  Our algorithm combines several ideas: 1) an alignment scheme for
  ``greedily'' fitting the output of the deletion channel as a
  subsequence of the input; 2) a version of the idea of ``anchoring''
  used in \cite{HMPW08}; and 3) complex analysis techniques from
  recent work of Nazarov and Peres \cite{NP16} and De, O'Donnell, and
  Servedio \cite{DOS16}.
\end{abstract}

\newpage

\section{Introduction}

The \emph{deletion channel} takes as input a bit string $\bx \in
\{0,1\}^n$. Each bit of $\bx$ is (independently of other bits)
retained with probability $p$ and deleted with probability $q := 1 -
p$. The channel then outputs the concatenation of the retained bits;
such an output is called a \emph{trace}. Suppose that the input $\bx$
is unknown. The \emph{trace reconstruction problem} asks the
following: how many i.i.d.\ traces from the deletion channel do we
need to observe in order to determine $\bx$ with high probability?

There are two basic variants of this problem, which we will call the
``worst case'' and ``average case''. In the worst case variant, the
problem is to provide bounds that hold uniformly over all possible
input strings $\bx$. The average case variant supposes that the input
is chosen uniformly at random. In particular, we are allowed to ignore
some ``hard-to-reconstruct'' inputs, as long as they comprise a small
fraction of all $2^n$ possible inputs. In this paper, we study the
average case. Our main result is the following.

\begin{theorem} \label{thm:subpoly-reconstruction}
  Suppose $q < \frac{1}{2}$, and let $\bX \in \{0,1\}^n$ be an unknown
  bit string of length $n$ chosen uniformly at random. There is a
  constant $C_q$ depending only on $q$ such that it is possible to
  reconstruct $\bX$ with probability at least $1 - \frac{C_q}{n}$ using
  at most $\exp\( C_q \sqrt{\log n}\)$ independent samples from the
  deletion channel with deletion probability $q$ applied to $\bX$.
\end{theorem}

\subsection{Related work}

The study of trace reconstruction for the deletion channel seems to
have been initiated by Batu, Kannan, Khanna and McGregor
\cite{BKKM04}, who were motivated by multiple sequence alignment
problems in computational biology. We focus on the regime where the
deletion probability $q$ is held constant as $n$ grows.

Previously, the best bound in the average case was due to
Holenstein, Mitzenmacher, Panigrahy and Wieder \cite{HMPW08}, who gave an
algorithm for reconstructing random inputs using  polynomially many traces
when $q$ is less than some small threshold $c$.\footnote{The threshold
  $c$ is not given explicitly in \cite{HMPW08}. It seems that by optimizing their
  methods we cannot achieve $c > 0.07$.} Theorem
\ref{thm:subpoly-reconstruction} improves on this result in two ways:
the number of traces is subpolynomial, and we extend the range of
allowed $q$ to the interval $(0,1/2)$.

In \cite{HMPW08} it is also shown that $e^{O(n^{1/2} \log n)}$ traces
suffice for reconstruction with high probability with worst case
input. This was recently improved by Nazarov-Peres \cite{NP16} and
De-O'Donnell-Servedio \cite{DOS16} (simultaneously and independently)
to $e^{O(n^{1/3})}$. Their techniques, which we use in Section
\ref{sec:shifted-reconstruction}, play an important role in our
proofs.

The question of whether the above bounds are optimal remains open. The
best lower bounds known are of order $\log^2 n$ (McGregor, Price and
Vorotnikova \cite{MPV14}) in the average case and order $n$ in the
worst case (\cite{BKKM04}).

Other settings for trace reconstruction include the case when $q
\rightarrow 0$ (\cite{BKKM04}), when insertions and substitutions are
allowed as well as deletions (\cite{KM05}, \cite{VS08}), or when the
strings are taken over an alphabet whose size grows with $n$
(\cite{MPV14}). For a more comprehensive review of the literature, we
refer readers to the introduction of \cite{DOS16} or the survey of
Mitzenmacher \cite{M09}.

\subsection{Outline of approach} \label{subsec:outline}

Let us give a high-level description of the algorithm used to prove
Theorem \ref{thm:subpoly-reconstruction}. Suppose that we have already
reconstructed the first $k$ bits of $\bX$, and
we consider a new trace $\wt{\bX}$. Roughly speaking, our goal is to carry out the following
steps:

\begin{figure}
  \centering
  \def\svgwidth{\columnwidth}
  \input{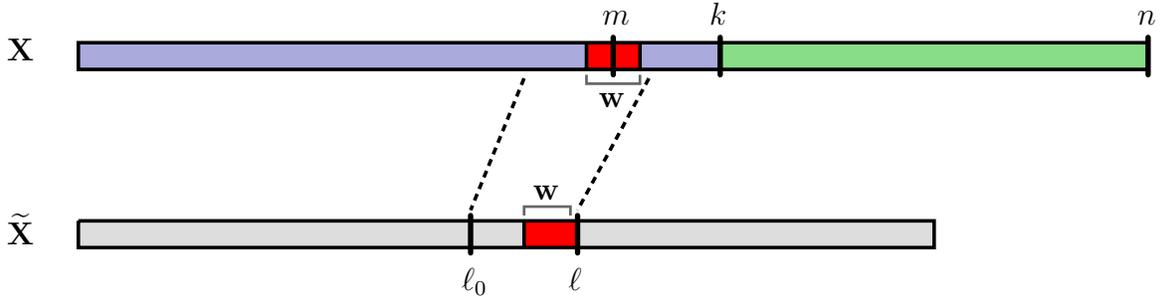}
  \caption{Illustration of the alignment strategy. Dotted lines
    indicate correspondences between positions in $\wt{\bX}$ and
    positions in $\bX$.}
  \label{fig:strategy-outline}
\end{figure}

\begin{enumerate}
\item[] {\bf Alignment:} Find some suitable index $m$ slightly less
  than $k$, and try to (approximately) identify the position $\ell$ in
  $\wt{\bX}$ that corresponds to the $m$-th position of $\bX$. This
  occurs in two stages (see Figure \ref{fig:strategy-outline}):
  \begin{enumerate}
  \item[] {\bf Initial alignment:} Find a position $\ell_0$ in
    $\wt{\bX}$ whose corresponding position in $\bX$ is known to be
    about $O(\log n)$ places ahead of $m$.
  \item[] {\bf Refined alignment:} Consider a specific substring $\bw$
    of $\bX$ located at $m$ and having length $O(\log^{1/2} n)$. Look
    for $\bw$ to occur in $\wt{\bX}$ within $O(\log n)$ characters
    following position $\ell_0$, and take $\ell$ to be the last
    position of this occurrence of $\bw$.
  \end{enumerate}
\item[] {\bf Reconstruction:} Use the bits of $\wt{\bX}$ after $\ell$
  as a trace of the bits of $\bX$ after $m$. From these ``traces'', we
  reconstruct at least $k + 1 - m$ bits of $\bX$ starting from position
  $m$, which in particular includes the $(k+1)$-th bit of $\bX$.
\end{enumerate}
We can repeat the above procedure for each $k$. In each iteration, the
number of traces needed will be $e^{O(\sqrt{\log n})}$. Moreover,
these traces may be reused for each iteration, because we will
ultimately bound the probability of failure by a union bound.

\subsubsection{Initial alignment step}

The initial alignment step is based on fitting $\wt{\bX}$ as a
subsequence of $\bX$ following a ``greedy algorithm''. Let $X_i$ and
$\wt{X}_i$ denote the $i$-th bits of $\bX$ and $\wt{\bX}$,
respectively. We associate $\wt{X}_1$ to the first bit in $\bX$ that
matches $\wt{X}_1$, then associate $\wt{X}_2$ to the next bit in $\bX$
that matches $\wt{X}_2$, and so on (see Figure
\ref{fig:greedy-illustration}). This gives the ``first possible''
occurrence of $\wt{\bX}$ as a subsequence of $\bX$, but does not
necessarily reflect the true alignment of $\wt{\bX}$ to
$\bX$. However, when $q < 1/2$ and $\bX$ is random, it turns out that
this greedy alignment actually matches the true one to within $O(\log
n)$ (stated precisely in Lemma \ref{lemma:trackability}).

Let us briefly describe why this is the case. Suppose that the
position assigned by our greedy algorithm lags behind the true
position. Looking at the next bit in the trace, the true position
should advance by $\frac{1}{1 - q} < 2$ places in
expectation. However, since the bits of $\bX$ are uniformly random,
the position for the greedy algorithm should advance like a geometric
random variable with mean $2$, thereby ``catching up''.

The same greedy matching idea was also considered by Mitzenmacher (see
Section 3 of \cite{M09}) in the slightly different context of decoding
for the deletion channel. Lemma \ref{lemma:trackability} is a
variant of Theorem 3.2 in \cite{M09}. However, many details are
omitted in \cite{M09}, so we provide a self-contained proof
in Section \ref{sec:align-greedy}.

\begin{figure}
  \centering
  \begin{picture}(200, 100)
    \put(0,0){\includegraphics[scale=3.0]{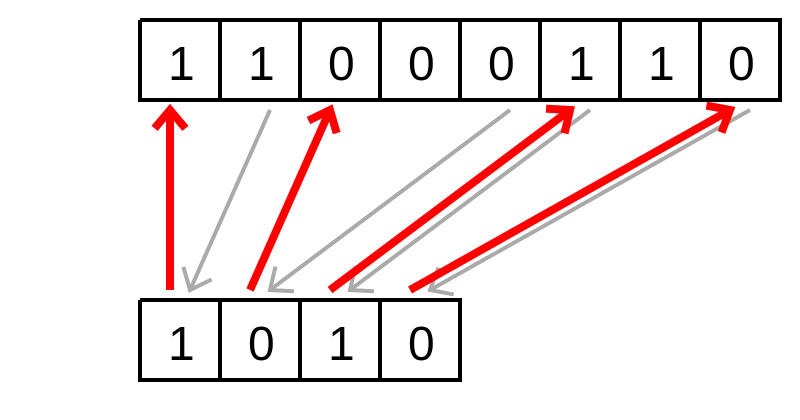}}
    \put(10,98){\resizebox{0.5cm}{!}{$\bX$:}}
    \put(10,18){\resizebox{0.5cm}{!}{$\wt{\bX}$:}}
  \end{picture}
  \caption{Illustration of the greedy algorithm used in the initial
    alignment step. Here, $\bX = 11000110$ and $\wt{\bX} = 1010$. Gray
    arrows point from the positions in $\bX$ that were retained to
    their corresponding positions in $\wt{\bX}$. Red arrows indicate
    the associations produced by our algorithm (i.e. $\wt{X}_1$ goes
    to $X_1$, $\wt{X}_2$ goes to $X_3$, $\wt{X}_3$ goes to $X_6$,
    $\wt{X}_4$ goes to $X_8$). }
  \label{fig:greedy-illustration}
\end{figure}

\subsubsection{Refined alignment step}

For the refined alignment, we take an approach similar to the use of
``anchors'' in \cite{HMPW08}. We again rely on the randomness of $\bX$
and the assumption $q < 1/2$. Consider a substring $\bw$ of length $a
\approx \log^{1/2} n$ which contains the $m$-th bit of $\bX$. (In the
language of \cite{HMPW08}, $\bw$ is our ``anchor''.)

With probability $p^a$, the string $\bw$ appears in our trace because
none of its bits were deleted. There is also a chance that this exact
sequence just happens to appear after deletions to another part of the
input. However, because $\bX$ is random, the latter scenario only
happens with probability $2^{-a} \ll p^a$. Thus, when we see $\bw$ in
our trace, it most likely came from near position $m$ of $\bX$ (we
discard traces if we do not see $\bw$), thereby aligning to within
$O(\log^{1/2} n)$.

We remark here that the above discussion sweeps under the rug a few
considerations about how to avoid accumulation of many small
probabilities of error. In particular, note that the error
probabilities involved during the refined alignment step are like
$e^{-O(\log^{1/2} n)}$, which is not small enough to union bound over
the whole string.

For example, a problem may arise if we have another copy of $\bw$
appearing in $\bX$ that is only $O(\log n)$ positions away from
$m$. In that case, appearances of $\bw$ in $\wt{\bX}$ might come from
either copy of $\bw$ in $\bX$, and it would be hard to distinguish the
two scenarios.

Recall, however, that we have allowed ourselves some flexibility in
the choice of $m$. Note that the initial alignment step means that we
only need to worry about what $\bX$ looks like within distance $O(\log
n)$ from the location $m$. We look at $O(\log^{1/2} n)$ possible
locations of $m$ which are spaced $O(\log n)$ apart, and we argue that
with high probability, at least one of these locations (and the
corresponding choice of $\bw$) behaves in the desired way.

\subsubsection{Reconstruction step}

For the reconstruction step, we analyze bit statistics using methods
based on those of \cite{NP16} and \cite{DOS16}. However, two
adaptations are needed for our setting. First, our reconstruction step
only needs to recover a small number of bits, not the full string. The
statement we need is roughly that $e^{O(r^{1/3})}$ traces are enough
to recover the first $r$ bits of an unknown string, which we apply
with $r = O(\log^{3/2} n)$.

Second, since our alignment is not perfect, we must allow some random
shifts of the input string. The amount of shifting we can tolerate is
relatively small, which explains the need for accurate alignment. The
issue of calculating bit statistics with random shifts also appears in
\cite{HMPW08}, although our techniques for handling this are rather
different from theirs.

These two adaptations can be carried out by small modifications to the
relevant proofs in \cite{NP16} and \cite{DOS16}, which are based on
bounds for Littlewood polynomials on arcs of the unit circle.

\subsection{Notation}

We will use boldface to denote bit strings, while the
values of their bits are non-bolded and subscripted by indices; for
example, $\bx = (x_1, x_2 \ldots , x_n) \in \{0,1\}^n$. Let $|\bx| =
n$ denote the length of $\bx$, and let $\bx^{a:b}$ denote the
substring $(x_a, x_{a+1}, \ldots , x_{b})$. For brevity, we also write
$\bx^{a:} = \bx^{a:|\bx|}$ for the suffix of $\bx$ starting at $x_a$.

Next, we introduce notation for describing the deletion channel. For a
given parameter $p \in (0, 1)$, let $\D^*_p(\bx)$ denote the
distribution over pairs $(\bt, \wt{\bx})$ of sequences defined as
follows: $\bt = (t_1, t_2, \ldots , t_m)$ is the random sequence of
indices of $\bx$ which are retained by the deletion channel applied to
$\bx$ with deletion probability $q = 1 - p$, and $\wt{\bx} =
(\wt{x}_1, \wt{x}_2, \ldots , \wt{x}_m)$ is given by $\wt{x}_i =
x_{t_i}$. Note that the length $m = |\bt|$ is random.

In some cases, we are only interested in the final output $\wt{\bx}$
and not in $\bt$. Thus, we also introduce the notation $\D_p(\bx)$ for
the marginal distribution of $\D^*_p(\bx)$ over the strings
$\wt{\bx}$. We will sometimes use the notation $\P_\bx(\,\cdot\,)$ to
emphasize that the string going through the deletion channel is $\bx$.

At some point, we will want to use $\bt$ to associate several indices
at once in $\wt{\bx}$ to their counterparts in $\bx$, or vice
versa. Consider sets $S \subseteq \{1, 2, \ldots , |\bx|\}$ and
$\wt{S} \subseteq \{1, 2, \ldots , |\wt{\bx}|\}$. Then, we use the
notation
\[ \bt(\wt{S}) := \{ t_s : s \in \wt{S} \} \qquad\text{and}\qquad \bt^{-1}(S) := \{ s : t_s \in S \}, \]
which matches the usual notation for images/preimages if $\bt$ is
regarded as a map from indices in $\wt{\bx}$ to indices in $\bx$.

Finally, in addition to the standard notation $O(\,\cdot\,)$
and $\Omega(\,\cdot\,)$, we also use $O_p(\,\cdot\,)$ and
$\Omega_p(\,\cdot\,)$ in cases where the implied constant may depend
on $p$ but nothing else.

\subsection{Organization of the paper}

The rest of the paper is organized as follows. In Sections
\ref{sec:align-greedy} and \ref{sec:align-seq}, we prove the lemmas
needed to for the initial and refined alignment steps,
respectively. In Section \ref{sec:shifted-reconstruction}, we prove
the lemmas needed for the reconstruction step. Finally, in Section
\ref{sec:main-proof}, we pull together all the ingredients to prove
Theorem \ref{thm:subpoly-reconstruction}.

\subsection*{Acknowledgements}

Most of this work was carried out while the second author was visiting
Microsoft Research in Redmond. He thanks Microsoft for the
hospitality.

\section{Alignment by greedy matching} \label{sec:align-greedy}

Suppose we have a string $\bx \in \{0,1\}^n$ and a sample $(\bt,
\wt{\bx}) \sim \D^*_p(\bx)$. Given only $\bx$ and $\wt{\bx}$, it is
not in general possible to infer uniquely what $\bt$ is. However, we
may obtain an approximation using a ``greedy algorithm'' as described
in Section \ref{subsec:outline}.

To state things precisely, consider any two bit strings $\bx$ and
$\by$. We define a sequence $(g_k(\by, \bx))_{k=1}^{|\by|}$ as
follows:
\begin{itemize}
\item Define $g_1(\by, \bx)$ to be the least index such that
  $x_{g_1(\by, \bx)} = y_1$. If no bits in $\bx$ are equal to $y_1$,
  we set $g_1(\by, \bx) = \infty$.
\item For $k < |\by|$, define inductively $g_{k + 1}(\by, \bx)$ to be
  the least index greater than $g_k(\by, \bx)$ for which
  $x_{g_{k+1}(\by, \bx)} = y_{k+1}$. If no bits in $\bx$ after the
  $g_k(\by, \bx)$-th position are equal to $y_{k+1}$, we set
  $g_{k+1}(\by, \bx) = \infty$. (Note that in particular if $g_k(\by,
  \bx) = \infty$, then $g_{k+1}(\by, \bx) = \infty$).
\end{itemize}

We are primarily interested in the case where $\by = \wt{\bx}$, where
$\wt{\bx}$ is a trace drawn from $\D_p(\bx)$. In this situation,
$g_k(\wt{\bx}, \bx)$ represents the ``earliest possible'' place in
$\bx$ that the $k$-th bit of $\wt{\bx}$ could have come from. For an
illustration, we refer back to Figure
\ref{fig:greedy-illustration}. In that picture, we have $g_1(\wt{\bx},
\bx) = 1$, $g_2(\wt{\bx}, \bx) = 3$, $g_3(\wt{\bx}, \bx) = 6$, and
$g_4(\wt{\bx}, \bx) = 8$.

One may check by a straightforward induction that $g_k(\wt{\bx}, \bx)
\le t_k$ for all $1 \le k \le |\wt{\bx}|$. (This means that
$g_k(\wt{\bx}, \bx)$ is never $\infty$; the possibility of having
$g_k(\by, \bx) = \infty$ doesn't come into play until the proof of
Lemma \ref{lemma:m-alignment}.)  We will show that for retention
probability $p > \frac{1}{2}$ and $\bx$ drawn uniformly at random,
$g_k(\wt{\bx}, \bx)$ is usually not much less than $t_k$. The
following definition makes this precise.

\begin{definition}
  Consider a sequence $\bx \in \{0,1\}^n$, and take $(\bt, \wt{\bx})
  \sim \D^*_p(\bx)$. We say that $\bx$ is \bemph{$(\alpha,
    \beta)$-trackable} if
  \[ \P_\bx\( \max_{1 \le k \le |\bt|} \( t_k - g_k(\wt{\bx}, \bx) \) \ge \lambda \) \le e^{-\frac{\lambda - \alpha}{\beta}}. \]
\end{definition}

The main result of this section is the following lemma.

\begin{lemma} \label{lemma:trackability}
  Suppose $p > \frac{1}{2}$, and let $\bX \in \{0,1\}^n$ be a
  uniformly random string of $n$ bits. There exists $C_p > 0$
  depending only on $p$ such that
  \[ \P\( \text{$\bX$ is $(C_p \log n, C_p)$-trackable} \) \ge 1 - O_p\(\frac{1}{n}\) \]
\end{lemma}

Lemma \ref{lemma:trackability} is implied by Theorem 3.2 of
\cite{M09}. However, many details are omitted there, so we devote the
rest of this section to proving Lemma \ref{lemma:trackability}
formally. We use the same general approach, except that it is more
natural for us to focus on the quantity $t_k - g_k(\wt{\bX}, \bX)$
rather than a slightly different quantity considered in
\cite{M09}. The starting point is a conditional independence property
similar to Lemma 3.3 of \cite{M09}.

\begin{lemma} \label{lemma:conditional-independence}
  Let $\bX \in \{0,1\}^n$ be drawn uniformly at random, and suppose
  $(\bt, \wt{\bX}) \sim \D^*_p(\bX)$. Then, for any integer $k \ge 1$,
  conditioned on the event $|\bt| \ge k$ and the values of
  \[ t_1, t_2, \ldots , t_k \text{ and } g_1(\wt{\bX}, \bX), g_2(\wt{\bX}, \bX), \ldots, g_k(\wt{\bX}, \bX), \]
  the bits $X_{g_k(\wt{\bX}, \bX) + 1}, X_{g_k(\wt{\bX}, \bX) + 2}, \ldots ,
  X_n$ are i.i.d.\ uniformly distributed.
\end{lemma}
\begin{remark}
  The above lemma also applies when $\bX$ is an infinite
  sequence of i.i.d.\ uniform bits. In this case, the conclusion is
  that all of $(X_i)_{i = g_k(\wt{\bX}, \bX) + 1}^\infty$ are
  i.i.d.\ uniform.
\end{remark}
\begin{proof}
  We first condition on $\bt$; this conditioning will stay in effect
  for the remainder of the proof. Note that all of the $X_i$ are still
  i.i.d.\ uniform, since the $t_i$ depend only on which bits are
  deleted and not on the values of the bits themselves. Since we have
  conditioned on $\bt$, we may regard $g_i(\wt{\bX}, \bX)$ as a
  deterministic function of $\bX$. Therefore, for brevity we will
  write $g_i(\bX) = g_i(\wt{\bX}, \bX)$.

  Next, fix any sequence $S$ of integers $s_1, s_2, \ldots , s_k$
  where $s_1 < s_2 < \cdots < s_k$ and $s_i \le t_i$ for each $i$. We
  say a bit string $\bz$ is \emph{$S$-compatible} if $g_i(\bz) =
  s_i$ for each $i$, and let $E_S$ be the event that $\bX$ is
  $S$-compatible.

  Consider any two strings $\bw, \bw' \in \{0,1\}^{n - s_k}$ which
  differ in a single bit. We will give a bijection between
  $S$-compatible realizations of $\bX$ that end in $\bw$ and those
  that end in $\bw'$. This is enough to establish the lemma, since by
  repeated application, it shows that any two strings for
  $\bX^{(s_k+1):}$ are equally likely conditioned on $E_S$, and this
  holds for arbitrary $S$.

  To carry out the bijection, for any index $j$ with $1 \le j \le k$,
  we define its \emph{influencing set} to be the set
  \[ I_j = \{ t : s_{j-1} < t \le s_j \}, \]
  with the convention $s_0 = 0$. Informally, it is the set of all
  indices $t$ where the value of $X_t$ had some effect on the value of
  $g_j(\bX)$ (which is equal to $s_j$ if $\bX$ is $S$-compatible).

  For any two indices $i$ and $j$ with $1 \le i, j \le k$, we say
  \emph{$i$ directly influences $j$} if $t_i \in I_j$. Note that
  because $s_j \le t_j$, we see that if $i$ influences $j$, then $i
  \le j$ with equality if and only if $s_i = t_i$. We say that
  \emph{$i$ influences $j$} if there is a chain of direct influences
  from $i$ to $j$ (i.e. there exist $c_1, c_2, \ldots , c_N$ such that
  $c_1 = i$, $c_N = j$, and $c_\alpha$ directly influences $c_{\alpha
    + 1}$ for $\alpha = 1, 2, \ldots , N - 1$).

  Suppose now that we have a $S$-compatible sequence $\bz$ that ends
  in $\bw$. We will describe a way to modify $\bz$ so that it remains
  $S$-compatible but ends in $\bw'$. Let $\ell$ be the index at which
  $w_\ell \ne w'_\ell$. First, suppose that $s_k + \ell \ne t_i$ for
  any $i \le k$. Then, we may simply flip the $(s_k + \ell)$-th bit of
  $\bz$ to obtain a $S$-compatible sequence ending in $\bw'$.

  Otherwise, $s_k + \ell = t_m$ for some $m \le k$. Define the sets
  \[ U = \{ m \} \cup \{ j : j \text{ influences } m \} \qquad\text{and}\qquad V = \{ t_m \} \cup \(\bigcup_{j \in U} I_j\). \]
  We claim that by flipping all the bits of $\bz$ at positions in $V$,
  the resulting sequence $\bz'$ ends in $\bw'$ and is
  $S$-compatible. The first claim follows from the fact that $I_j
  \subseteq \{1, 2, \ldots , s_k\}$ for all $j \le k$, so the only bit
  flipped after position $s_k$ is the bit at position $t_m = s_k +
  \ell$.

  To show $S$-compatibility, we show by induction on $j$ that
  $g_j(\bz') = s_j$ for each $j$, where the base case $j = 0$ is
  established by the convention $g_0(\bz') = s_0 = 0$. For the
  inductive step, suppose that $g_i(\bz') = s_i$ for each $i < j$. We
  consider two cases.
  \\\\
  \noindent{\bf Case $j \in U$.} By the definition of $U$, either $j =
  m$ or there exists $j' \in U$ for which $t_j \in I_{j'}$. In either
  case, we see that $t_j \in V$. We also have by definition that $I_j
  \subseteq V$. By $S$-compatibility of $\bz$, the condition $g_j(\bz) =
  s_j$ says that $s_j$ is the first position after $g_{j-1}(\bz) =
  s_{j-1}$ having the same value as position $t_j$. In other words,
  $s_j$ is the unique position in $I_j$ with the same value as
  position $t_j$.

  The bits at positions $t_j$ and elements of $I_j$ are all flipped
  for $\bz'$, so the same property holds in $\bz'$. Since $g_{j-1}(\bz') =
  s_{j-1}$ by the inductive hypothesis, we have $g_j(\bz') = s_j$ as
  well.
  \\\\
  \noindent{\bf Case $j \not\in U$.} Note that $t_m > s_k$, so $t_m
  \not\in I_j$. Since $j \not\in U$, it follows that $I_j$ is disjoint
  from $V$. Note that if $t_j \in I_{j'}$ for some $j' \in U$, then
  $j$ directly influences $j'$ and hence influences $m$, but this
  contradicts $j \not\in U$. Also, clearly $t_j \ne t_m$ since $j \ne
  m$. Thus, $t_j \not\in V$.

  We see that none of the bits at positions $t_j$ or elements of $I_j$
  are flipped for $\bz'$, so by the same argument as in the previous
  case, we conclude that $g_j(\bz') = s_j$.
  \\\\
  This completes the induction, showing that $\bz'$ indeed ends in
  $\bw'$ and is $S$-compatible. Furthermore, observe that the set $V$
  depends only on $S$, and so we may symmetrically recover $\bz$ from
  $\bz'$ by the same transformation. Thus, this gives a bijection from
  $S$-compatible sequences ending in $\bw$ to those ending in $\bw'$,
  completing the proof.
\end{proof}

The next two lemmas describe how closely $g_k$ tracks $t_k$. To avoid
boundary issues, it is convenient to state them for infinite bit
sequences.

\begin{lemma} \label{lemma:geometric-increments}
  Let $\bX$ be an infinite sequence of i.i.d.\ uniform bits, and let
  $(\bt, \wt{\bX}) \sim \D^*_p(\bX)$. Define $d_k = t_k -
  g_k(\wt{\bX}, \bX)$.

  Then, $d_{k+1} - d_k$ is independent of $d_1, d_2, \ldots , d_k$ and
  has the same law as $\max(G_p - G_{1/2}, -d_k)$, where $G_p$ and
  $G_{1/2}$ are independent geometrics with parameters $p$ and
  $\frac{1}{2}$, respectively.
\end{lemma}
\begin{proof}
  For brevity, write $g_k = g_k(\wt{\bX}, \bX)$. We condition on $t_i$
  and $g_i$ for $1 \le i \le k$. By Lemma
  \ref{lemma:conditional-independence}, the bits $(X_i)_{i = g_k +
    1}^\infty$ are i.i.d.\ uniform even after this conditioning. Next,
  we sample $t_{k+1}$, which we may write as $t_{k+1} = t_k + G_p$
  since each bit is retained independently with probability $p$. We
  then examine the bits
  \[ X_{g_k + 1}, X_{g_k + 2}, \ldots , X_{t_{k + 1}}, \]
  which are still i.i.d.\ uniformly distributed. Recall that $g_{k+1}$
  is defined to be the earliest position of these bits where the value
  matches $\wt{X}_{k+1} = X_{t_{k+1}}$. Each of the above bits has a
  $\frac{1}{2}$ chance of being a match except for the last one, which
  is guaranteed to match.

  Thus, $g_{k+1}$ may be written as $\min(g_k + G_{1/2},
  t_{k+1})$. Consequently,
  \[ d_{k+1} - d_k = (t_{k + 1} - t_k) - (g_{k + 1} - g_k) = G_p - \min(G_{1/2}, t_{k+1} - g_k) \]
  \[ = \max(G_p - G_{1/2}, G_p + g_k - t_{k+1}) = \max(G_p - G_{1/2}, -d_k), \]
  as desired.
\end{proof}

\begin{lemma} \label{lemma:d_k-deviation}
  Suppose $p > \frac{1}{2}$. Let $\bX$ be an infinite sequence of
  i.i.d.\ uniform bits, and let $(\bt, \wt{\bX}) \sim
  \D^*_p(\bx)$. Define $d_k = t_k - g_k(\wt{\bX}, \bX)$.

  Then, there exist positive constants $c_p$ and $C_p$ depending only
  on $p$ such that for each $k$, we have
  \[ \P\( d_k \ge \lambda \) \le C_pe^{-c_p \lambda}. \]
\end{lemma}
\begin{proof}
  Let $G_p$ and $G_{1/2}$ be independent geometrics with parameters
  $p$ and $\frac{1}{2}$, as in Lemma
  \ref{lemma:geometric-increments}. Consider the function $f(x) =
  \frac{px}{1 - (1-p)x} \cdot \frac{1}{2x - 1}$, which satisfies $f(1)
  = 1$ and $f'(1) = \frac{1}{p} - 2 < 0$. Thus, we may take $\alpha
  \in (1, 1/p)$ to be a constant so that $f(\alpha) < 1$. We will show
  that $\E(\alpha^{d_k})$ is bounded above uniformly in $k$, from
  which the result immediately follows by Markov's inequality.

  We proceed by induction. For the base case, note that $t_1$ has the
  distribution of $G_p$, and we chose $\alpha < 1/p$, so
  $\E(\alpha^{G_p})$ is finite. Since $d_1 \le t_1$, it follows that
  $\E(\alpha^{d_1})$ is also finite.

  For the inductive step, define $\gamma = f(\alpha) < 1$, and let $M$
  be a large enough integer so that
  \[ 1 + 2 (\alpha - 1) \(\frac{1}{2\alpha}\)^M \le \frac{1}{\sqrt{\gamma}}. \]
  Note that we have the formulas
  \begin{align*}
    \E\( \alpha^{G_p} \) &= p \sum_{k = 1}^\infty (1 - p)^{k - 1} \alpha^k = \frac{p\alpha}{1 - (1 - p)\alpha} \\
    &= \\
    \E\( \alpha^{-\min(G_{1/2}, M)} \) &= 2^{-M} \alpha^{-M} + \sum_{k = 1}^M 2^{-k} \alpha^{-k} = \(\frac{1}{2\alpha}\)^M + \frac{1}{2\alpha} \cdot \frac{1 - \(\frac{1}{2\alpha}\)^M}{1 - \frac{1}{2\alpha}} \\
    &= \frac{1 + 2(\alpha - 1)\(\frac{1}{2\alpha}\)^M}{2\alpha - 1} \le \frac{1}{\sqrt{\gamma}(2\alpha - 1)}.
  \end{align*}
  These calculations allow us to bound two conditional expectations,
  depending on whether $d_k \ge M$. By Lemma
  \ref{lemma:geometric-increments}, we have
  \begin{align*}
    \E\( \alpha^{d_{k+1}} \mid d_k \ge M \) &\le \E\(\alpha^{\max(G_p - G_{1/2}, -M)}\) \E\( \alpha^{d_k} \mid d_k \ge M \) \\
    &\le \E\( \alpha^{G_p} \) \E\( \alpha^{-\min(G_{1/2}, M)} \) \E\( \alpha^{d_k} \mid d_k \ge M \) \\
    &\le \frac{p\alpha}{1 - (1 - p)\alpha} \cdot \frac{1}{\sqrt{\gamma}(2\alpha - 1)} \cdot \E\( \alpha^{d_k} \mid d_k \ge M \) \\
    &\le \sqrt{\gamma} \cdot \E\( \alpha^{d_k} \mid d_k \ge M \) \\
    \E\( \alpha^{d_{k+1}} \mid d_k < M \) &\le \E\(\alpha^{M + G_p}\) = \frac{p\alpha^{M + 1}}{1 - (1 - p)\alpha} \le 8\alpha^M.
  \end{align*}
  Together, these imply that
  \[ \E\( \alpha^{d_{k+1}} \) \le \sqrt{\gamma} \cdot \E\( \alpha^{d_k} \) + 8\alpha^M. \]
  Recall that $\alpha$, $\gamma$, and $M$ are all constants that
  depend only on $p$, and $\gamma < 1$. Hence, $\E\( \alpha^{d_k} \)$
  is bounded above uniformly in $k$, completing the proof.
\end{proof}

We are finally ready to prove Lemma \ref{lemma:trackability}.

\begin{proof}[Proof of Lemma \ref{lemma:trackability}]
  For a given string $\bz$ of $n$ bits and $(\bt_\bz, \wt{\bz}) \sim
  \D^*_p(\bz)$, write
  \[ d(\bz) = \max_{1 \le k \le |\bt_\bz|} \( t_{\bz,k} - g_k(\wt{\bz}, \bz) \) \qquad\text{and}\qquad r_\lambda(\bz) = \P_\bz\( d(\bz) \ge \lambda \). \]
  We apply Lemma \ref{lemma:d_k-deviation} to the sequence $\bX$,
  where we may think of $\bX$ as the first $n$ bits of an infinite
  sequence of i.i.d.\ uniform bits. Union bounding over all indices $1
  \le k \le n$, we have
  \[ \E[r_\lambda(\bX)] \le n \cdot C_{1,p} \cdot e^{-c_{1,p} \lambda}, \]
  where $C_{1,p}$ and $c_{1,p}$ are constants depending only on
  $p$. Consequently,
  \begin{equation} \label{eq:r_lambda-bound}
    \P\( r_\lambda(\bX) \ge e^{-c_{1,p}\lambda/2} \) \le n \cdot C_{1,p} \cdot e^{-c_{1,p}\lambda/2}.
  \end{equation}
  Define the event
  \[ E = \bigcap_{\lambda = 2\ceil{\log n}}^\infty \Big\{ r_{2\lambda/c_{1,p}}(\bX) \le e^{-\lambda} \Big\}. \]
  Then, a union bound using \eqref{eq:r_lambda-bound} gives
  \begin{equation} \label{eq:E-bound}
    \P(E) \ge 1 - n \cdot C_{1,p} \sum_{\lambda = 2\ceil{\log n}}^\infty e^{-\lambda} \ge 1 - \frac{C_{2,p}}{n},
  \end{equation}
  where $C_{2,p}$ is another constant depending only on $p$.

  Meanwhile, on the event $E$, consider any $t > \frac{2}{c_{1,p}}(2\ceil{\log n}
  + 1)$. Let $t' = \left\lfloor \frac{c_{1,p} t}{2}
  \right\rfloor$. Since $t' \ge 2\ceil{\log n}$, we have
  \begin{align}
    \P\( d(\bX) \ge t \) &\le \P\( d(\bX) \ge \frac{2 t'}{c_{1,p}} \) = r_{2t'/c_{1,p}}(\bX) \le e^{-t'} \nonumber \\
    &\le e^{- \frac{c_{1,p} t}{2} + 1}. \label{eq:d(X)-bound}
  \end{align}
  Combining \eqref{eq:E-bound} and \eqref{eq:d(X)-bound}, we conclude
  that
  \[ \P\( \text{$\bX$ is $(C_p \log n, C_p)$-trackable} \) \ge 1 - \frac{C_p}{n} \]
  for a sufficiently large constant $C_p$.
\end{proof}

\section{Alignment by seeing a particular sequence} \label{sec:align-seq}

In this section, we develop the tools for our second alignment
strategy based on looking for a particular sequence of consecutive
bits. The strategy follows the same main idea as the use of
``anchors'' in \cite{HMPW08}. However, our analysis is more
precise. We first establish some terminology and notation.

\begin{definition}
  For any two bit strings $\bw$ and $\by$, we say that \bemph{$\bw$
    occurs in $\by$} if there is some index $j$ such that $y_{j + i -
    1} = w_i$ for $i = 1, 2, \ldots , |\bw|$. We use the following
  notation to describe occurrences:

  \begin{itemize}
  \item $\ind_\bw(\by)$ denotes the first index at which $\bw$ occurs
    in $\by$ (i.e. the smallest possible $j$ as above), or $\infty$ if
    $\bw$ does not occur in $\by$.
  \item Whenever $\ind_\bw(\by) < \infty$,
    \[ \indset_\bw(\by) := \{ j : \ind_\bw(\by) \le j < \ind_\bw(\bw) + |\bw| \} \]
    denotes the set of all the indices in $\by$ corresponding to the
    occurrence of $\bw$ in $\by$.
  \end{itemize}

  In later sections, we will be interested in occurrences of $\bw$
  within a particular substring $\by^{i:j}$ of $\by$. However, we
  still want to work with indices based on position in $\by$ rather
  than in $\by^{i:j}$. In these cases, we use the notation

  \begin{itemize}
  \item $\ind^{i:j}_\bw(\by) := \ind_\bw(\by^{i:j}) + i - 1$.
  \item $\indset^{i:j}_\bw(\by) := \{ k : \ind^{i:j}_\bw(\by) \le k < \ind^{i:j}_\bw(\bw) + |\bw| \}$.
  \end{itemize}
\end{definition}

Suppose that $\bx$ is a string of length $2n$, and $\bw =
\bx^{(n-a+1):(n+a)}$ is a substring in the middle of $\bx$. Now,
suppose we observe a trace $\wt{\bx} \sim \D_p(\bx)$, and we see that
$\bw$ occurs in $\wt{\bx}$. We would like to say that in this case the
bits in $\wt{\bx}$ corresponding to the occurrence of $\bw$ likely
came from the occurrence of $\bw$ in $\bx$ (or at least, some of them
did). Not all strings $\bx$ have this property, but as we will see
shortly, it turns out that typical ones do. We formalize the property
in the following definition.

\begin{definition} \label{def:distinguishable-center}
  Suppose $p > \frac{1}{2}$, let $\bx \in \{0,1\}^{2n}$, and take
  $(\bt, \wt{\bx}) \sim \D^*_p(\bx)$. Consider a positive integer $a
  \le n$ and positive real $\gamma < 1$, and write $\bw =
  \bx^{(n-a+1):(n+a)}$. We say that $\bx$ is \bemph{$(a,
    \gamma)$-distinguishable} if
  \[ \P_\bx\Big( \ind_\bw(\wt{\bx}) < \infty \quad\text{and}\quad \bt(\indset_\bw(\wt{\bx})) \cap [n-a, n+a] = \emptyset \Big) \le \gamma^a \cdot p^{2a}. \]
\end{definition}
\begin{remark}
  It is always possible for $\bw$ to occur in $\wt{\bx}$ if each of
  the positions $n-a+1$ through $n+a$ in $\bx$ are retained. This
  happens with probability $p^{2a}$. The bound on the probability in
  the above definition is given in the form $\gamma^a \cdot p^{2a}$ to
  highlight that it should be smaller than $p^{2a}$ by a factor that
  is exponential in $a$.
\end{remark}

The main result of this section is that random sequences are likely to
be distinguishable.

\begin{lemma} \label{lemma:distinguish-centers}
  Suppose $p > \frac{1}{2}$, and suppose $\bX \in \{0,1\}^{2n}$ is
  chosen uniformly at random. Then, there exist $\gamma_p < 1$ and
  $c_p > 0$ depending only on $p$ such that
  \[ \P\( \text{$\bX$ is $(\ceil{n^{1/2}}, \gamma_p)$-distinguishable} \) \ge 1 - e^{-c_p n^{1/2}}. \]
\end{lemma}
\begin{proof}
  Let $a = \ceil{n^{1/2}}$, let $\bw = \bX^{(n-a+1):(n+a)}$, and take
  $(\bt, \wt{\bX}) \sim \D^*_p(\bX)$. Let
  \[ J = \bt^{-1}\Big([1, 2n] \setminus [n-a, n+a]\Big) = \Big\{ j : 1 \le j \le |\wt{\bX}|, \quad t_j \not\in [n-a, n+a] \Big\} \]
  denote the set of indices of $\wt{\bX}$ which did not come from the
  middle $2a$ positions of $\bX$. Define the event
  \[ E = \Big\{ \ind_\bw(\wt{\bX}) < \infty \quad\text{and}\quad \bt(\indset_\bw(\wt{\bX})) \subseteq J \Big\}, \]
  which is the relevant event for $(a, \gamma)$-distinguishability.

  Let us condition on the middle $2a$ bits of $\bX$ (i.e. the bits
  that form $\bw$) as well as on $\bt$. The key observation is that
  $(\wt{X}_j)_{j \in J}$ are still i.i.d.\ uniform after our
  conditioning. Now, if $\bw$ occurs in $\wt{\bX}$, but
  $\bt(\indset_\bw(\wt{\bX})) \subseteq J$, then it means that $\bw$
  occurs in the sequence $(\wt{X}_j)_{\text{$j \in J$}}$. However,
  since the $(\wt{X}_j)_{j\in J}$ are i.i.d., in each possible
  position this only happens with probability $2^{-|\bw|} =
  2^{-2a}$. Union bounding over at most $2n$ positions yields
  \[ \P(E) \le 2n \cdot 2^{-2a}, \]
  where we have also taken the expectation over our initial
  conditioning on the middle $2a$ bits and $\bt$.

  The above probability is with respect to simultaneously two sources
  of randomness: the random choice of $\bX$ and the random choice of
  the deletions. To highlight this, recall the notation $\P_\bx$ for
  the probability over the randomness of the deletion channel for a
  given input string $\bx$.

  Take $\gamma_p = (2p)^{-1/2} < 1$. By Markov's inequality,
  \begin{align*}
    \P(\P_\bX(E) \ge \gamma_p^a \cdot p^{2a}) &\le \gamma_p^{-a} \cdot p^{-2a} \cdot \E(\P_\bX(E)) = \gamma_p^{-a} \cdot p^{-2a} \cdot \P(E) \\
     &\le 2n \cdot \gamma_p^{3a} = e^{-\Omega_p(n^{1/2})},
  \end{align*}
  which yields $(a, \gamma_p)$-distinguishability with the desired probability.
\end{proof}

We conclude the section by establishing a consequence of
$(\ceil{n^{1/2}}, \gamma_p)$-distinguishability that is more
convenient to work with than Definition
\ref{def:distinguishable-center}.

\begin{lemma} \label{lemma:distinguishability-conseqeunce}
  Suppose $p > \frac{1}{2}$, let $a = \ceil{n^{1/2}}$, and suppose
  $\bx \in \{0,1\}^{2n}$ is $(a, \gamma_p)$-distinguishable for some
  constant $\gamma_p < 1$ depending only on $p$. Consider $(\bt, \wt{\bx})
  \sim \D^*_p(\bx)$. Then,
  \[ \P_\bx\Big( \ind_\bw(\wt{\bx}) < \infty \quad\text{and}\quad \bt(\indset_\bw(\wt{\bx})) \not\subseteq [n-10a, n+10a] \Big) \le e^{-\Omega_p(a)} \cdot p^{2a}. \]
\end{lemma}
\begin{proof}
  The main idea is that if the set $\bt(\indset_\bw(\wt{\bx}))$
  intersects the interval $[n-a, n+a]$, then it is unlikely to stretch
  out very far from that interval.

  Let
  \[ E_1 = \Big\{ \ind_\bw(\wt{\bx}) < \infty \quad\text{and}\quad \bt(\indset_\bw(\wt{\bx})) \cap [n-a, n+a] = \emptyset \Big\}, \]
  so that $(a, \gamma_p)$-distinguishability ensures $\P_\bx(E_1) \le
  \gamma_p^a p^{2a}$.

  Next, let
  \[ E_2 = \left\{ \begin{tabular}{c} \text{more than $7a$ deletions occurred among} \\ \text{some $9a$ consecutive positions in $\bx$} \end{tabular} \right\}. \]
  By a standard Chernoff bound (see, e.g., \cite{hoeffding}) and union
  bounding over all blocks of $9a$ bits in $\bx$, we have
  \begin{align*}
    \P_\bx(E_2) &\le 2n \cdot \P\( \Binom(9a, 1/2) > 7a \) \le 2n \cdot e^{- \frac{25a^2}{18a}} \\
    &\le 2n \cdot 4^{-a} \le e^{-\Omega_p(a)} \cdot p^{2a}.
  \end{align*}

  Finally, let
  \[ E_3 = \Big\{ \ind_\bw(\wt{\bx}) < \infty \;\text{and}\; \bt(\indset_\bw(\wt{\bx})) \not\subseteq [n-10a, n+10a] \Big\}, \]
  which is the event of interest for the lemma. Suppose now that $E_1$
  holds but not $E_3$, i.e. $\bt(\indset_\bw(\wt{\bx}))$ is not disjoint
  from $[n - a, n + a]$ but is also not contained within $[n - 10a, n
    + 10a]$. Then $\bt(\indset_\bw(\wt{\bx}))$ must have two elements which
  are at least $9a$ apart, so that $E_2$ holds. Thus, we find that
  \[ \P_\bx(E_3) \le \P_\bx(E_1) + \P_\bx(E_2) \le e^{-\Omega_p(a)} \cdot p^{2a}. \]
\end{proof}

\section{Reconstruction from approximate alignment} \label{sec:shifted-reconstruction}

In this section, we adapt the trace reconstruction methods of
\cite{NP16} and \cite{DOS16} to a setting where the input string also
undergoes a random shift. The main result of this section is the
following lemma.

\begin{lemma} \label{lemma:distinguishing-with-shifts}
  Let $k$, $n$, and $N$ be positive integers with $k < n < N$. Let
  $\bx, \bx' \in \{0,1\}^N$ be two strings whose first $k$ digits are
  identical but whose first $n$ digits are not. Let $S$ be a random
  variable taking integer values between $0$ and $k - 1$.

  Suppose the following conditions are satisfied:
  \[ \E[ |S - \E S| ] \le n^{1/3}, \qquad k \le n^{2/3}. \]
  Then, for some constant $C_p$ depending only $p$, there exists an
  index $j \le C_p n$ such that if $\wt{\bx} \sim \D_p(\bx^{(S+1):})$ and
  $\wt{\bx}' \sim \D_p((\bx')^{(S+1):})$, then
  \[ \left| \P_\bx(\wt{x}_j = 1) - \P_{\bx'}(\wt{x}'_j = 1) \right| \ge \exp\( -C_p n^{1/3} \). \]
\end{lemma}

The first ingredient in the proof of this lemma is a polynomial
identity, which is analogous to Lemma 2.1 in \cite{NP16} or Section 4
in \cite{DOS16}, but accounts for possible shifts to the input
sequence.

\begin{lemma} \label{lemma:poly-identity}
  Let $n$ and $k$ be positive integers with $k \le n$. Let $\ba =
  (a_0, a_1, \ldots, a_{n - 1})$ be a sequence of real numbers whose
  first $k$ elements are zero. Let $S$ be a random variable taking
  integer values between $0$ and $k - 1$, with $\P(S = i) = \beta_i$.

  Let $\wt{\ba} \sim \D_p(\ba^{(S+1):})$, and pad $\wt{\ba}$ with
  zeroes to the right. Then,
  \begin{equation}
    \E \left[ \sum_{j \geq 0} \wt{a}_{j} w^{j} \right] = p \( \sum_{s = 0}^{k - 1} \beta_s (pw + q)^{-s} \) \( \sum_{j=0}^{n-1} a_{j} \left( pw+q \right)^{j} \).
  \end{equation}
\end{lemma}
\begin{proof}
  Let $\ell$ be any integer with $k \le \ell \le n - 1$. By linearity,
  it suffices to show the result for $\ba$ having all zeroes except
  $a_\ell = 1$. We now restrict to this case.

  Let us condition on $S = s$ and analyze for each $j$ the probability
  $\P(\wt{a}_j = 1)$ that the single non-zero entry $a_\ell$ gets
  shifted to position $j$ without being deleted. Clearly, if $j > \ell
  - s$, then $\P(\wt{a}_j = 1 \mid S = s) = 0$. Otherwise, the
  probability must account for the retention of $a_\ell$ and the
  retention of exactly $j$ of the first $\ell - s$ entries of
  $\ba^{(s+1):}$. Note that the condition $k \le \ell$ ensures that
  $\ell - s > 0$. Thus for $j \le \ell-s$,
  \[ \P\( \wt{a}_j = 1 \mid S = s \) = p \cdot \binom{\ell - s}{j} p^j q^{\ell - s - j}, \]
  so that
  \[ \E\( \sum_{j \geq 0} \wt{a}_{j} w^{j} \,\middle|\, S = s \) = p \sum_{j = 0}^{\ell - s} \binom{\ell - s}{j} p^j q^{\ell - s - j} w^j = p \cdot (pw + q)^{\ell - s}. \]
  Taking the expectation over $S$, we conclude that
  \begin{align*}
    \E\( \sum_{j \geq 0} \wt{a}_{j} w^{j} \) &= \sum_{s = 0}^{k - 1} \beta_s \E\( \sum_{j \geq 0} \wt{a}_{j} w^{j} \,\middle|\, S = s \) = p \sum_{s = 0}^{k - 1} \beta_s (pw + q)^{\ell - s} \\
    &= p \( \sum_{s = 0}^{k - 1} \beta_s (pw + q)^{-s} \) (pw + q)^\ell,
  \end{align*}
  which completes the proof.
\end{proof}

As in \cite{NP16} and \cite{DOS16}, we also use the following
Littlewood-type estimate of Borwein and Erd{\'e}lyi.

\begin{lemma}[Borwein and Erd{\'e}lyi, special case of Corollary 3.2 in \cite{BE97}] \label{lemma:polynomial-bound}
  There exists a finite constant $C$ such that the following holds.
  Let $A(z)$ be a polynomial with coefficients in $[-1, 1]$ and $A(0)
  = 1$. Denote by $\gamma_L$ the arc $\left\{ e^{i \theta} : -1/L \leq
  \theta \leq 1/L \right\}$. Then $\max_{z \in \gamma_L} |A(z)| \geq
  e^{-CL}$.
\end{lemma}

We now carry out the proof of Lemma
\ref{lemma:distinguishing-with-shifts} using these two ingredients.

\begin{proof}[Proof of Lemma \ref{lemma:distinguishing-with-shifts}]
  For a fixed value of $p$, clearly it is enough to prove the
  statement for sufficiently large $n$. We will assume implicitly at
  various points that $n$ is sufficiently large.

  Write $\beta_j = \P(S = j)$, let $a_j = x_{j+1} - x'_{j+1}$, and
  let $\ba = (a_j)_{j = 0}^{n - 1}$. Define the polynomials
  \[ P(z) = \sum_{j = 0}^{k - 1} \beta_j z^j, \quad Q(z) = \sum_{j = 0}^{n - 1} a_j z^j, \quad\text{and}\quad A(z) = p \cdot P(z^{-1})Q(z). \]
  Let $\ell$ be the smallest index for which $a_\ell \ne 0$; note that
  by our hypotheses, $\ell \le n$. Define $\wt{Q}(z) =
  \frac{1}{z^\ell} Q(z)$, so that $|\wt{Q}(0)| = 1$.

  For convenience, let $L = n^{1/3}$, and define $\rho = 1 -
  1/L^2$. Applying Lemma \ref{lemma:polynomial-bound} to the function
  $\wt{Q}(\rho z)$, there exists $z_0 = e^{i\theta}$ with $-
  \frac{p}{10L} \le \theta \le \frac{p}{10L}$ and $|\wt{Q}(\rho z_0)|
  \ge e^{-CL/p}$.

  We next lower bound $|P(\rho^{-1} z_0^{-1})|$. Let $\wt{P}(z) =
  z^{-\E S} P(z)$, which is an analytic function on the right
  half-plane. For all $z$ in the right half-plane satisfying $1 \le |z| \le
  \rho^{-1}$, differentiating $\wt{P}$ gives
  \begin{align*}
    |\wt{P}'(z)| &= \left| \sum_{j = 0}^{k - 1} (j - \E S) \beta_j z^{j - \E S - 1} \right| \le \sum_{j = 0}^{k - 1} |j - \E S| \cdot |z|^{j - \E S - 1} \\
    &\le \rho^{-k} \cdot \E[ |S - \E S| ] \le \rho^{-k} L \le e^{\frac{1.1k}{L^2}} \cdot L \le 4L,
  \end{align*}
  where we have used $\E[ |S - \E S| ] \le L$ and $k \le L^2$. Also,
  \begin{align*}
    |\rho^{-1} z_0^{-1} - 1| &= \rho^{-1} |1 - \rho z_0| \le |z_0 - 1| + \rho^{-1}(1 - \rho) \\
    &\le \frac{p}{10L} + \frac{2}{L^2} \le \frac{p}{8L}.
  \end{align*}
  Consequently,
  \begin{align*}
    |P(\rho^{-1}z_0^{-1})| &= \rho^{-\E S} |\wt{P}(\rho^{-1}z_0^{-1})| \ge 1 - |\wt{P}(\rho^{-1}z_0^{-1}) - 1| \\
    &= 1 - \left| \int_1^{\rho^{-1}z_0^{-1}} \wt{P}'(z) \,dz \right| \ge 1 - |\rho^{-1}z_0^{-1} - 1| \cdot 4L \\
    &\ge 1 - \frac{p}{2} \ge \frac{1}{2}.
  \end{align*}
  Thus,
  \[ |A(\rho z_0)| = p \cdot |P(\rho^{-1}z_0^{-1})| \cdot \rho^\ell \cdot |\wt{Q}(\rho z_0)| \ge \frac{p}{2} \cdot e^{- \frac{1.1n}{L^2} - \frac{CL}{p}} \ge e^{- \frac{(C + 2)L}{p}}. \]

  Next, define $w = 1 + \frac{\rho z_0 - 1}{p}$, so that $\rho z_0 =
  pw + q$. We have that
  \begin{align*}
    |w|^2 &= 1 + \frac{2}{p}(\rho \cdot \Re(z_0) - 1) + \frac{1}{p^2}|\rho z_0 - 1|^2 \\
    &\le 1 + \frac{2}{p}\(\rho - 1\) + \frac{\rho^2}{p^2} |\rho^{-1} z_0^{-1} - 1|^2 \\
    &\le 1 - \frac{2}{L^2} + \frac{1}{64L^2} \le \rho.
  \end{align*}
  Let $\wt{\ba} \sim \D_p(\ba^{(S+1):})$. By Lemma
  \ref{lemma:poly-identity},
  \[ \left| \E \left[ \sum_{j \geq 0} \wt{a}_{j} w^{j} \right] \right| = |A(\rho z_0)| \ge e^{- \frac{(C + 2)L}{p}}. \]
  Now, take $C_p$ to be an integer larger than $\frac{C + 4}{p}$. Note
  that
  \[ \left| \sum_{j = C_p n}^\infty \E[\wt{a}_j] w^j \right| \le \sum_{j = C_p n}^\infty \rho^j \le L^2 \rho^{C_p n} \le \frac{1}{2} \cdot e^{- \frac{(C + 2)L}{p}}. \]
  Hence,
  \[ \left| \E \left[ \sum_{j = 0}^{C_p n - 1} \wt{a}_{j} w^{j} \right] \right| \ge \frac{1}{2} \cdot e^{- \frac{(C + 2)L}{p}} \ge e^{-(C_p - 1)L}, \]
  and therefore, we must have for some $j$ with $0 \le j \le C_p n - 1$ that
  \[ \left| \P(\wt{x}_j = 1) - \P(\wt{x}'_j = 1) \right| = |\E \wt{a}_j| \ge |\E \wt{a}_j w^j| \]
  \[ \ge \frac{1}{C_p n} e^{-(C_p - 1)L} \ge e^{-C_pL}, \]
  as desired.
\end{proof}

\section{Proof of Theorem \ref{thm:subpoly-reconstruction}} \label{sec:main-proof}

Throughout this section, we fix a deletion probability $q <
\frac{1}{2}$ (and hence a retention probability $p > \frac{1}{2}$). In
addition, all of our inequalities are meant to apply for $n$
sufficiently large (i.e. larger than a constant depending only on
$p$).

Let $C_p$ be the larger of the two constants in Lemmas
\ref{lemma:trackability} and \ref{lemma:distinguishing-with-shifts},
and let $c_p$ be the constant in Lemma
\ref{lemma:distinguish-centers}. We define the following integers:
\[ M = \ceil{C_p \log n}, \qquad K_1 = 40M, \qquad K_0 = \ceil{K_1^{1/2}}, \qquad K_2 = \ceil{ \frac{10}{c_p} K_1^{1/2}\log n }. \]
It is helpful to keep in mind that $K_0 = \Theta_p(\log^{1/2} n)$,
$K_1 = \Theta_p(\log n)$, and $K_2 = \Theta_p(\log^{3/2} n)$.

Recall the high-level strategy of the proof from Section
\ref{subsec:outline}: we align traces against what we have
reconstructed so far, and then we use bit statistics to reconstruct
additional bits. The alignment step in particular relies on the input
$\bX$ having certain special properties which don't hold for all
strings but do hold for ``most''. We encapsulate these properties in
the following definition.

\begin{definition} \label{defgood}
  Let $\gamma_p < 1$ be the constant from Lemma
  \ref{lemma:distinguish-centers}. We say that a string $\bx \in
  \{0,1\}^n$ is \bemph{good} if the following conditions are
  satisfied:
  \begin{enumerate}
  \item[(i).] $\bx$ is $(M, C_p)$-trackable,
  \item[(ii).] there is no run of $M$ consecutive identical bits in $\bx$,
  \item[(iii).] among any $K_2$ consecutive bits of $\bx$, there is a
    block of $2K_1$ of them that is $(K_0, \gamma_p)$-distinguishable.
  \end{enumerate}
\end{definition}

As the next lemma shows, a random string is good with high probability.

\begin{lemma} \label{lemma:most-good}
  Let $\bX \in \{0,1\}^n$ be drawn uniformly at random. Then,
  \[ \P(\text{$\bX$ is good}) \ge 1 - O_p\(\frac{1}{n}\). \]
\end{lemma}
\begin{proof}
  It suffices to show that each condition in  \ref{defgood} holds with probability at
  least $1 - O_p\(\frac{1}{n}\)$. For condition (i), this is immediate
  by Lemma \ref{lemma:trackability}.

  To establish condition (ii), note that the probability for $M$
  i.i.d.\ uniform bits to be identical is $2^{1-M}$. Union bounding
  over all blocks of $M$ consecutive bits in $\bX$, we find that (ii)
  holds with probability at least $1 - n \cdot 2^{1-M} \ge 1 -
  O\(\frac{1}{n}\)$.

  Finally, for condition (iii), note that any $K_2$ consecutive bits
  contain at least $\floor{K_2/2K_1} \ge \frac{2 \log n}{c_p \sqrt{K_1}}$ disjoint blocks of size
  $2K_1$. By Lemma \ref{lemma:distinguish-centers}, the probability
  that a single block fails to be $(K_0, \gamma_p)$-distinguishable is
  at most $e^{-c_p\sqrt{K_1}}$. Thus, the probability that
  none of these blocks is $(K_0, \gamma_p)$-distinguishable is at most
  \[ \exp\( -\frac{2 \log n}{c_p \sqrt{K_1}} \cdot c_p \sqrt{K_1} \) = \frac{1}{n^2}. \]
  Union bounding over at most $n$ possible blocks of
  $K_2$ consecutive bits shows that condition (iii) also holds with
  probability at least $1 - O\(\frac{1}{n}\)$.
\end{proof}

\subsection{Alignment}

Suppose $\bx$ is a bit string that we know, and let $m \le |\bx|$ be
some position in $\bx$. Suppose that we also have a sample $\wt{\bx}$
from the deletion channel applied to $\bx$ (or some longer string
having $\bx$ as a prefix). As described in Section
\ref{subsec:outline}, we would like to identify (with high
probability) a bit of $\wt{\bx}$ that was originally positioned near
the $m$-th bit of $\bx$. This motivates the following definition.

\begin{definition} \label{def:alignment-rule}
  An \bemph{alignment rule} is a function $\L$ which takes as input a
  bit string $\bx$, an index $m \le |\bx|$, and another bit string
  $\by$. It outputs a value $\L(\bx, m, \by) \in \{ 1, 2, \ldots ,
  |\by| - 1, |\by|, \infty \}$.

  In addition, we require that $\L$ satisfy the following adaptedness
  property with respect to $\by$: whenever $\L(\bx, m, \by) < \infty$,
  for any other string $\by'$ identical to $\by$ in their first
  $\L(\bx, m, \by)$ bits, we have $\L(\bx, m, \by') = \L(\bx, m,
  \by)$.
\end{definition}

Let us explain the conceptual meaning of this definition. Although it
is not strictly required for the definition, we emphasize that for our
purposes, $\by$ will be a sample from the deletion channel applied to
a string whose prefix is $\bx$. The idea is that bits near the $m$-th
position of $\bx$ should end up near the $\L(\bx, m, \by)$-th position
in $\by$ after going through the deletion channel; in this way, the
position $m$ in $\bx$ is ``aligned'' with position $\L(\bx, m, \by)$
in $\by$. When $\L(\bx, m, \by) = \infty$, it means that the rule
cannot reliably locate which bits of $\by$ came from around the $m$-th
position of $\bx$.

The adaptedness condition says that an alignment rule must proceed by
examining the bits of $\by$ in order one by one, and either outputting
the current position or giving up and outputting $\infty$. In
particular, we do not allow alignment rules to look ahead in the
string $\by$ before deciding whether a previous position should be the
output. The purpose of this requirement is to ensure that the deletion
pattern after our alignment position is independent of the alignment
itself.

The next lemma constructs a particular alignment rule that has good
quantitative bounds on the quality of the alignment.

\begin{lemma} \label{lemma:m-alignment}
  Let $k$ be a given integer with $K_2 \le k \le n/2$, and let $\bx_0
  \in \{0,1\}^k$ be a string of length $k$. Then, there exists an
  index $m$ with $k - K_2 + K_1 \le m \le k - K_1$ and an alignment
  rule $\L$ with the following property:

  For any good sequence $\bx \in \{0, 1\}^n$ with $\bx_0$ as a prefix,
  taking $(\bt, \wt{\bx}) \sim \D^*_p(\bx)$, we have
  \begin{enumerate}
  \item[(i).] $\P_\bx(\L(\bx_0, m, \wt{\bx}) < \infty) \ge
    \frac{1}{2} p^{2K_0}$
  \item[(ii).] $\P_\bx(|t_{\L(\bx_0, m, \wt{\bx})} - m| \ge K_1 \mid
    \L(\bx_0, m, \wt{\bx}) < \infty) \le n^{-\Omega(1)}$
  \item[(iii).] $\P_\bx(|t_{\L(\bx_0, m, \wt{\bx})} - m| \ge 10K_0 \mid
    \L(\bx_0, m, \wt{\bx}) < \infty) \le e^{-\Omega_p(K_0)}$.
  \end{enumerate}
\end{lemma}

Informally speaking, the properties in the above lemma should be
interpreted as saying that (i) the alignment succeeds with some
not-too-small probability; (ii) it is extremely likely to align within
$K_1$ of the correct position; and (iii) it usually aligns within
$10K_0$. Before giving the proof, we first establish an auxiliary lemma.

\begin{lemma} \label{lemma:greedy-alignment-bound}
  Suppose $\bx \in \{0,1\}^n$ is a good sequence, and suppose $(\bt,
  \wt{\bx}) \sim \D^*_p(\bx)$. Consider any $k \le n/2$, and let
  $\ell$ be the smallest index such that $g_\ell(\wt{\bx}, \bx) \ge
  k$. Then,
  \[ \P_\bx\Big(\text{$\ell$ exists and } k \le t_\ell \le k + 4M \Big) \ge 1 - \frac{1}{n} \]
  for all sufficiently large $n$.
\end{lemma}
\begin{proof}
  Let $n' = \floor{\frac{5}{6} \cdot pn}$. If $\ell$ does not exist, it
  means that $g_{|\wt{\bx}|}(\wt{\bx}, \bx) < k$ (or $\wt{\bx}$ is
  empty). This can be bounded by
  \begin{align}
    &\P_\bx(\text{$\ell$ does not exist}) \le \P_\bx(|\wt{\bx}| < n') + \P_\bx(g_{n'}(\wt{\bx}, \bx) < n/2) \nonumber \\
    &\qquad\qquad\le \P_\bx(|\wt{\bx}| < n') + \P_\bx(t_{n'} - g_{n'}(\wt{\bx}, \bx) \ge n/6) + \P_\bx(t_{n'} < 2n/3). \label{eq:ell-existence1}
  \end{align}
  Note that $|\wt{\bx}|$ is distributed as $\Binom(n, p)$, so
  $\P_\bx(|\wt{\bx}| < n') = e^{-\Omega_p(n)}$. The second term in
  \eqref{eq:ell-existence1} is at most $e^{-\Omega(n)}$ because $\bx$
  was assumed to be $(M, C_p)$-trackable. Finally, if $t_{n'} < 2n/3$,
  it means that at least $5pn/6$ out of the first $2n/3$ bits were
  retained, which also occurs with probability at most
  $e^{-\Omega_p(n)}$. Thus, all three probabilities in
  \eqref{eq:ell-existence1} are exponentially small in $n$, so
  \begin{equation} \label{eq:ell-existence2}
    \P_\bx(\text{$\ell$ does not exist}) \le \frac{1}{n^2}
  \end{equation}
  for large enough $n$.

  We now work under the assumption that $\ell$ exists. We always have
  \begin{equation} \label{eq:t_ell-bound1}
    t_\ell \ge g_\ell(\wt{\bx}, \bx) \ge k
  \end{equation}
  Since $\bx$ does not have more than $M$ consecutive identical bits,
  by the minimality of $\ell$, we must have $g_\ell(\wt{\bx}, \bx) \le
  k + M$. By $(M, C_p)$-trackability of $\bx$, we have
  \begin{equation} \label{eq:t_ell-bound2}
    \P_\bx(t_\ell > k + 4M) \le \P_\bx(t_\ell - g_\ell(\wt{\bx}, \bx) > 3M) \le \frac{1}{n^2}.
  \end{equation}
  Combining \eqref{eq:ell-existence2}, \eqref{eq:t_ell-bound1}, and
  \eqref{eq:t_ell-bound2} completes the proof.
\end{proof}

\begin{proof}[Proof of Lemma \ref{lemma:m-alignment}]
  If $\bx_0$ is not a prefix of any good sequence, then there is
  nothing to prove. Otherwise, because $\bx_0$ is a prefix of a good
  sequence, there must exist $m$ with $k - K_2 + K_1 \le m \le k -
  K_1$ such that $\bx_0^{(m-K_1+1):(m+K_1)}$ is $(K_0,
  \gamma_p)$-distinguishable. We choose such an $m$ and let $\bw =
  \bx_0^{(m-K_0+1):(m+K_0)}$ (see Figure \ref{fig:m-alignment}).

  Now, suppose $\bx$ is any good sequence having $\bx_0$ as a prefix,
  and take $(\bt, \wt{\bx}) \sim \D^*_p(\bx)$. Roughly speaking, our
  alignment rule will be to first use Lemma
  \ref{lemma:greedy-alignment-bound} to identify an index $\ell_0$ in
  $\wt{\bx}$ such that $t_{\ell_0}$ is slightly smaller than
  $m$. Then, we will look for an occurrence of $\bw$ in $\wt{\bx}$
  shortly after position $\ell_0$. If such an occurrence exists, we
  output the position of the last bit of the occurrence. If not, we
  output $\infty$.

  \begin{figure}
    \centering
    \def\svgwidth{\columnwidth}
    \input{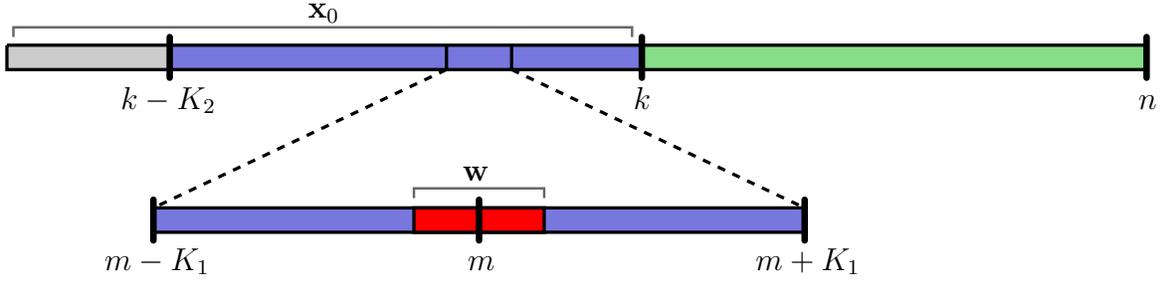}
    \caption{Illustration of positions involved in the proof of Lemma \ref{lemma:m-alignment}.}
    \label{fig:m-alignment}
  \end{figure}

  To specify the alignment rule precisely, consider any string $\by$,
  and define the statement
  \[ P_0(\by) = \text{``there exists $\ell$ such that $\infty > g_{\ell}(\by, \bx_0) \ge m - 8M$''}.\footnote{Recall that since $\by$ need not be drawn from $\D_p(\bx_0)$,
    it is possible to have $g_\ell(\by, \bx_0) = \infty$.} \] Whenever
  $P_0(\by)$ holds, take $\ell_0(\by)$ to be the smallest such
  $\ell$. Then, define
  \[ P_1(\by) = P_0(\by) \wedge \text{``$\ind^{\ell_0(\by):(\ell_0(\by)+16M)}_\bw(\by) < \infty$''}. \]
  We then define our alignment rule to be
  \[ \L(\bx_0, m, \by) = \begin{cases*}
    \ind^{\ell_0(\by):(\ell_0(\by)+16M)}_\bw(\by) + 2K_0 - 1 & if $P_1(\by)$ holds \\
    \infty & otherwise.
    \end{cases*} \]
  Note that this satisfies the adaptedness requirement for alignment
  rules. We will specifically apply the above definition with $\by =
  \wt{\bx}$, so it is convenient to define the events
  \[ E_0 = \{ \text{$P_0(\wt{\bx})$ holds} \}, \qquad E_1 = \{ \text{$P_1(\wt{\bx})$ holds} \}, \]
  and we abbreviate $\ell_0 = \ell_0(\wt{\bx})$.

  Next, we establish properties (i), (ii), and (iii). In what follows,
  the reader may find it helpful to refer to Figure
  \ref{fig:m-alignment2}. Define
  \[ F_0 = E_0 \cap \{ m - 8M \le t_{\ell_0} \le m - 4M \}. \]
  By Lemma \ref{lemma:greedy-alignment-bound}, we have
  \begin{equation} \label{eq:F0}
    \P_\bx(F_0) \ge 1 - \frac{1}{n}
  \end{equation}
  We note a subtlety in our use of the lemma: the event $E_0$ concerns
  existence of $g_\ell(\wt{\bx}, \bx_0)$, while Lemma
  \ref{lemma:greedy-alignment-bound} concerns existence of
  $g_\ell(\wt{\bx}, \bx)$. However, as long as $t_{\ell_0} \le m -
  4M$, the relevant indices are all less than $k$, so there is no
  difference between using $\bx_0$ and using $\bx$, and the lemma
  still applies.

  We can now lower bound $\P_\bx(\L(\bx_0, m, \wt{\bx}) < \infty) =
  \P_\bx(E_1)$. Conditioned on $F_0$, it is always possible for $E_1$
  to occur by retaining all the bits in positions $m-K_0+1$ through
  $m+K_0$ in $\bx$. Thus,
  \[ \P_\bx(\L(\bx_0, m, \wt{\bx}) < \infty) = \P_\bx(E_1) \ge \P_\bx(F_0) \cdot p^{2K_0} \ge \frac{1}{2} p^{2K_0}, \]
  establishing property (i).

  \begin{figure}
    \centering
    \def\svgwidth{\columnwidth}
    \input{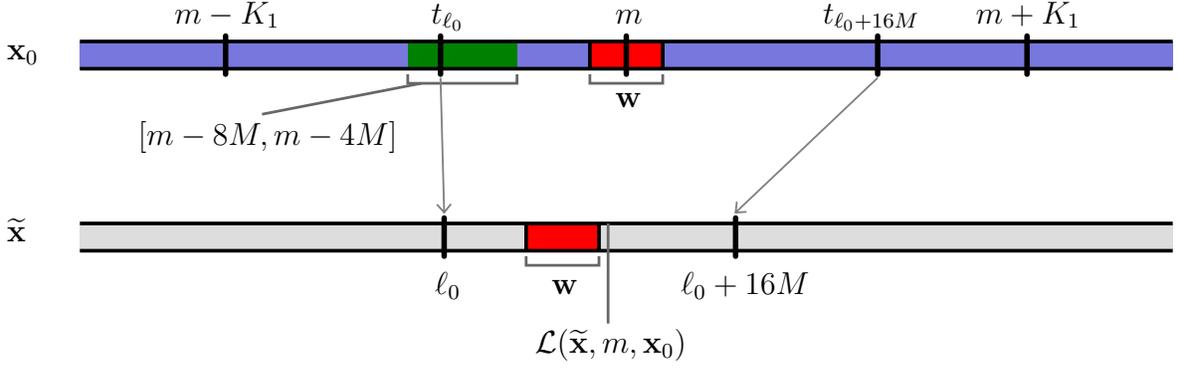}
    \caption{A possible configuration for $\bx_0$, $m$, and
      $\wt{\bx}$. In the diagram above, events $F_0$, $F_1$, and $E_1$
      all hold.}
    \label{fig:m-alignment2}
  \end{figure}

  To show property (ii), consider the event
  \[ F_1 = \{ t_{\ell_0 + 16M} \le m + K_1 \}. \]
  Note that if $F_0$ and $F_1^c$ both occur, then it means that fewer
  than $16M$ bits were retained among the positions in $\bx$ between
  $m - 4M$ and $m + K_1$. There are $K_1 + 4M = 44M$ such positions,
  so
  \begin{equation} \label{eq:F1}
    \P_\bx(F_0 \cap F_1^c) \le \P_\bx(\Binom(44M, p) < 16M) \le e^{-\Omega(M)} = n^{-\Omega(1)}.
  \end{equation}
  If $F_0$, $F_1$, and $E_1$ all occur, then we have
  \begin{align*}
    t_{\L(\bx, m, \wt{\bx})} &\ge t_{\ell_0} \ge m - 8M \ge m - K_1 \\
    t_{\L(\bx, m, \wt{\bx})} &\le t_{\ell_0 + 16M} \le m + K_1.
  \end{align*}
  Thus,
  \begin{align*}
    &\P_\bx\Big(|t_{\L(\bx, m, \wt{\bx})} - m| \le K_1 \;\Big|\; E_1 \Big) \ge \P_\bx(F_0 \cap F_1 \mid E_1) \ge 1 - \frac{\P_\bx(F_0^c) + \P_\bx(F_0 \cap F_1^c)}{\P_\bx(E_1)} \\
    &\qquad\qquad \ge 1 - n^{-\Omega(1)},
  \end{align*}
  establishing property (ii).

  Finally, we show property (iii). Let
  \[ I = \bt^{-1}\( \{m-K_1+1, m-K_1+2, \ldots , m+K_1\} \) \]
  be the set of indices in $\wt{\bx}$ which ``came from''
  $\bx_0^{(m-K_1+1):(m+K_1)}$. Note that we can regard
  $(\wt{x}_i)_{\text{$i \in I$ in increasing order}}$ as being drawn
  from $\D_p\(\bx_0^{(m-K_1+1):(m+K_1)}\)$. Consider the event
  \[ F_2 = E_1 \cap \left\{ \bt\(\indset^{\ell_0:(\ell_0+16M)}_\bw(\wt{\bx})\) \subseteq [m - 10K_0, m + 10K_0] \right\}. \]
  Note that we have the implication
  \[ \bt\(\indset^{\ell_0:(\ell_0+16M)}_\bw(\wt{\bx})\) \not\subseteq [m - 10K_0, m + 10K_0] \qquad\text{and}\qquad [\ell_0 : (\ell_0 + 16M)] \subseteq I \]
  \[ \implies \bt\(\indset^I_\bw(\wt{\bx})\) \not\subseteq [m - 10K_0, m + 10K_0], \]
  which means
  \begin{align}
    \P_\bx\(F_2^c \cap (F_0 \cap F_1 \cap E_1)\) &\le \P_\bx\Big( \bt\(\indset^I_\bw(\wt{\bx})\) \not\subseteq [m - 10K_0, m + 10K_0] \Big) \nonumber \\
    &\le e^{-\Omega_p(K_0)} \cdot p^{2K_0}, \label{eq:F2}
  \end{align}
  where the last inequality follows from the fact that
  $\bx_0^{(m-K_1+1):(m+K_1)}$ is $(K_0, \gamma_p)$-distinguishable
  combined with Lemma \ref{lemma:distinguishability-conseqeunce}.

  Recall from property (i) that $\P_\bx(E_1) \ge \frac{1}{2} p^{2K_0}$. We
  conclude that
  \begin{align*}
    &\P_\bx\Big(|t_{\L(\wt{\bx})} - m| \le 10K_0 \;\Big|\; E_1 \Big) \ge \P_\bx(F_0 \cap F_2 \mid E_1) \\
    &\qquad\qquad \ge 1 - \frac{\P_\bx(F_0^c) + \P_\bx(E_1 \cap F_0 \cap F_1^c) + \P_\bx(E_1 \cap F_0 \cap F_1 \cap F_2^c)}{\P_\bx(E_1)} \\
    &\qquad\qquad \ge 1 - n^{-\Omega(1)} - n^{-\Omega(1)} - e^{-\Omega_p(K_0)} = 1 - e^{-\Omega_p(K_0)},
  \end{align*}
  where we have used \eqref{eq:F0}, \eqref{eq:F1}, and \eqref{eq:F2}
  to bound the numerator appearing in the second line. This proves
  (iii).
\end{proof}

\subsection{Reconstruction}

The following lemma provides a template for how we will reconstruct
bits.

\begin{lemma} \label{lemma:reconstruction-template}
  Consider integers $k_1$ and $k_2$ with $k_1 < k_2$, and let
  $\mathcal{S} \subseteq \{0,1\}^{k_2}$ be a known set of length-$k_2$
  bit strings. Suppose that we have a number $\epsilon > 0$ and a
  family of statistics $b_j : \mathcal{S} \to \R$ for $1 \le j \le
  k_2$ which satisfies the following property: for any two strings
  $\bw, \bw' \in \mathcal{S}$ whose first $k_1$ bits are not
  identical, there exists an index $j_{\bw,\bw'}$ such that
  $|b_{j_{\bw,\bw'}}(\bw) - b_{j_{\bw,\bw'}}(\bw')| > \epsilon$.

  Let $\bz \in \mathcal{S}$ be an unknown string, and suppose that we
  observe estimates $(\hat{b}_j)_{j = 1}^{k_2}$ such that $|\hat{b}_j
  - b_j(\bz)| < \epsilon/2$ for each $j$. Then, we can determine the
  first $k_1$ bits of $\bz$.
\end{lemma}
\begin{proof}
  For any two strings $\bw, \bw' \in \mathcal{S}$ whose first $k_1$
  bits are not identical, we say that $\bw$ \emph{beats} $\bw'$ if
  $\hat{b}_{j_{\bw,\bw'}}$ is closer to $b_{j_{\bw,\bw'}}(\bw)$ than
  to $b_{j_{\bw,\bw'}}(\bw')$. We say $\bw$ is \emph{dominant} if it
  beats all other strings $\bw' \in \mathcal{S}$ that do not share its
  first $k_1$ bits.

  Our hypotheses imply that $\bz$ is dominant. Moreover, any two
  dominant strings must share their first $k_1$ bits. Thus, we may
  recover the first $k_1$ bits of $\bz$ as the first $k_1$ bits of any
  dominant string.
\end{proof}

We now apply the template in two lemmas. The first lemma shows how to
reconstruct the initial $K_2$ bits, and the second lemma shows how to
reconstruct additional bits once we have already reconstructed a long
enough prefix of $\bx$.

\begin{lemma} \label{lemma:initial-segment}
  Let $\bx \in \{0,1\}^n$ be a good sequence. There is a constant
  $C'_p$ depending only on $p$ such that $N = \ceil{\exp\(C'_p
      \sqrt{\log n}\)}$ independent samples from $\D_p(\bx)$ are
  sufficient to recover the first $K_2$ bits of $\bx$ with probability
  at least $1 - \frac{1}{n}$ for all sufficiently large $n$.
\end{lemma}
\begin{proof}
  \def\xavg{\wt{x}^{\text{avg}}}

  Let $\wt{\bx}_1, \ldots , \wt{\bx}_N$ be the sampled traces. For
  each $j \le n$, let
  \[ \xavg_j = \frac{1}{N} \sum_{i = 1}^N \wt{x}_{i,j} \]
  be the average of the bits of the $\wt{\bx}_i$ at position $j$,
  where $\wt{\bx}_i$ are padded to the right with zeroes.

  We will apply Lemma \ref{lemma:reconstruction-template} with $k_1 =
  K_2$ and $k_2 = n$. We consider statistics $b_j(\bz)$ equal to the
  expected value of the $j$-th bit of a string drawn from
  $\D_p(\bz)$. By Lemma \ref{lemma:distinguishing-with-shifts}, we may
  take $\epsilon = e^{-O_p\(K_2^{1/3}\)} = e^{-O_p\(\log^{1/2} n\)}$.

  Choose $C'_p$ sufficiently large so that $\epsilon^2 N \ge
  e^{\sqrt{\log n}}$. Noting that $\E[\xavg_j] = b_j(\bx)$, by a
  Chernoff bound we have
  \[ \P_\bx( |\xavg_j - b_j(\bx)| > \epsilon/2) \le e^{-\frac{\epsilon^2 N}{2}} \le \frac{1}{n^2} \]
  for all large enough $n$. Thus, letting
  \[ E = \{ \text{$|\xavg_j - b_j(\bx)| \le \epsilon/2$ for each $j$} \} \]
  and union bounding over all $1 \le j \le n$, we have $\P_\bx(E) \ge 1 -
  \frac{1}{n}$.

  Using $\hat{b}_j = \xavg_j$ as our estimates, Lemma
  \ref{lemma:reconstruction-template} asserts that we can recover the
  first $K_2$ bits of $\bx$ on the event $E$, which proves the desired
  statement.
\end{proof}

\begin{lemma} \label{lemma:advance-bit}
  Let $n$ be a positive integer, and let $k$ be an integer with $K_2
  \le k \le n/2$. There is a constant $C'_p$ depending only on $p$
  such that the following holds:

  Consider a good sequence $\bx \in \{0,1\}^n$, and suppose that $N :=
  \ceil{\exp\(C'_p \sqrt{\log n}\)}$ i.i.d.\ samples $\wt{\bx}_1,
  \ldots , \wt{\bx}_N$ are drawn from $\D_p(\bx)$. Then, whenever $n$
  is sufficiently large, seeing only the first $k$ bits of $\bx$ and
  the traces $\wt{\bx}_1, \ldots , \wt{\bx}_N$ is sufficient to
  recover the $(k + 1)$-th bit of $\bx$ with probability at least $1 -
  \frac{1}{n^2}$.
\end{lemma}
\begin{proof}
  \def\bavg{\hat{b}^{\text{avg}}}

  Let $m$ and $\L$ be the index and alignment rule given by Lemma
  \ref{lemma:m-alignment}, where we take $\bx_0 = \bx^{1:k}$ to be the
  first $k$ bits of $\bx$ (which we have been given). Let us consider
  a single trace $\wt{\bx} \sim \D_p(\bx)$. For brevity, write $\ell =
  \ell(\wt{\bx}) = \L(\bx_0, m, \wt{\bx})$.

  We say that $\wt{\bx}$ is a \emph{usable} trace if $\ell <
  \infty$. Let $E$ denote the event that $\wt{\bx}$ is usable, and let
  \[ E' = E \cap \{ m - K_1 \le t_{\ell} \le m + K_1 \} \]
  \[ E'' = E \cap \{ m - 10K_0 \le t_{\ell} \le m + 10K_0 \}. \]
  Lemma \ref{lemma:m-alignment} ensures that $\P_\bx(E' \mid E) \ge 1 -
  n^{-\Omega(1)}$ and $\P_\bx(E'' \mid E) \ge 1 - e^{-\Omega_p(K_0)}$,
  which taken together imply that
  \begin{equation} \label{eq:E''|E'}
    \P_\bx(E'' \mid E') \ge 1 - e^{-\Omega_p(K_0)} = 1 - e^{-\Omega_p(\log^{1/2} n)}.
  \end{equation}

  Let $H = m - K_1$, and let $\Delta$ be a random variable having the
  same distribution as $t_{\ell} - H$ conditioned on $E'$. The reason
  for defining $\Delta$ in this particular way will be made clearer
  shortly. For now, let us take note of several properties of
  $\Delta$:
  \begin{itemize}
  \item $\Delta$ is an integer between $0$ and $2K_1$.
  \item The distribution of $\Delta$ can be calculated just by looking
    at $\bx_0$ (in particular, it does not depend on bits of $\bx$
    after the $k$-th one).\footnote{It should be noted that the
      probability $\P_\bx(E)$ of having a usable trace \textit{does}
      depend on later bits of $\bx$. However, the additional
      constraint $m - K_1 \le t_{\ell} \le m + K_1$ combined with the
      adaptedness property of $\L$ removes this dependence.}
  \item By \eqref{eq:E''|E'}, it is straightforward to deduce that $\E
    \Delta = m + O_p(K_0)$ and $\E[|\Delta - \E\Delta|] = O_p(K_0)$.
  \end{itemize}
  Define $K_3 = \ceil{C''_p \log^{3/2} n}$, where $C''_p$ is a
  large enough constant to ensure that
  \[ \E[|\Delta - \E\Delta|] \le K_3^{1/3}, \quad 2K_1 \le K_3^{2/3}, \quad\text{and}\quad K_3 > K_2. \]
  Our goal will be to distinguish the true suffix $\bx^{(H+1):}$ from
  other possible suffixes via Lemma
  \ref{lemma:reconstruction-template}, where we take $(k_1, k_2) =
  (K_3, n - H)$. Here, the set $\mathcal{S}$ is taken to be all
  strings of length $n - H$ having $\bx^{(H+1):k}$ as a prefix. By
  reconstructing the first $K_3$ bits of $\bx^{(H+1):}$, we will have
  in particular reconstructed $x_{k+1}$, since
  \[ H + K_3 = m - K_1 + K_3 > k + 1. \]
  The statistics we use are, for any $\bz \in \mathcal{S}$,
  \[ b_j(\bz) := \text{expected value of the $j$-th bit of a string drawn from
    $\D_p(\bz^{(\Delta+1):})$}, \] and we take $\epsilon = e^{-C_p
    K_3^{1/3}}$ (it may be helpful to keep in mind that $\epsilon =
  e^{-\Theta_p(\log^{1/2} n)}$). Note that we are able to compute
  these quantities $b_j(\bz)$ since we are able to compute the
  distribution of $\Delta$.

  Let us first verify the property required of the $b_j$ and
  $\epsilon$ in Lemma \ref{lemma:reconstruction-template}. Consider
  any two strings $\bw, \bw' \in \mathcal{S}$ that do not agree in
  their first $K_3$ bits. We apply Lemma
  \ref{lemma:distinguishing-with-shifts} to these strings with $(k, n,
  S) = (2K_1, K_3, \Delta)$. To check the hypotheses of the lemma,
  note that by the definition of $\mathcal{S}$ and the assumption $m
  \le k - K_1$, $\bw$ and $\bw'$ agree in their first $k - H = k + K_1
  - m \ge 2K_1$ bits, as required. We also recall that by the way we
  defined $K_3$, the conditions
  \[ \E[|\Delta - \E\Delta|] = O_p(K_0) \le K_3^{1/3}, \qquad 2K_1 \le K_3^{2/3} \]
  are satisfied. Thus, Lemma \ref{lemma:distinguishing-with-shifts}
  tells us that there exists an index $j_{\bw,\bw'}$ for which
  \[ |b_j(\bw) - b_j(\bw')| \ge e^{-C_p K_3^{1/3}} = \epsilon, \]
  establishing that our choice of $b_j$ and $\epsilon$ are suitable
  for use in Lemma \ref{lemma:reconstruction-template}.

  Unfortunately, we cannot directly observe samples with the law of
  $\D_p(\bx^{(H+1+\Delta):})$ in order to estimate
  $b_j(\bx^{(H+1):})$. However, a usable trace $\wt{\bx}$ allows us to
  sample from this distribution approximately. The fact that $\L$ is
  adapted to $\wt{\bx}$ (as required in Definition
  \ref{def:alignment-rule}) means that if we condition on
  $t_{\ell(\wt{\bx})} = h$ for some index $h$, the string
  $\wt{\bx}^{(\ell+1):}$ has the same distribution as
  $\D_p(\bx^{(h+1):})$. Thus, the definition of $\Delta$ ensures that,
  conditioned on the event $E'$, $\wt{\bx}^{(\ell+1):}$ has exactly
  the law of $\D_p(\bx^{(H+1+\Delta):})$.

  As long as $\wt{\bx}$ is usable, we define $\hat{b}_j(\wt{\bx}) :=
  \wt{\bx}_{\ell+j}$ and $\overline{b}_j = \E(\hat{b}_j(\wt{\bx}) \mid
  E)$. The above discussion implies that
  \begin{equation} \label{eq:b-bar}
    \left| \overline{b}_j - b_j(\bx^{(H+1):}) \right| \le \P_\bx(E'^c \mid E) \le n^{-\Omega(1)} \le \epsilon/4,
  \end{equation}
  where the bound on $\P_\bx(E'^c \mid E)$ comes from Lemma
  \ref{lemma:m-alignment}.

  Averaging over our $N$ traces $\wt{\bx}_1, \ldots , \wt{\bx}_N$ will
  then give us a fairly good estimate on $b_j(\bx^{(H+1):})$. Choose
  $C'_p$ large enough so that the following hold:
  \begin{align}
    N \ge 64 p^{-6K_0} &\implies \frac{1}{2} p^{2K_0} N \ge 2N^{2/3} \label{eq:N-def1} \\
    N \ge \epsilon^{-3} e^{\sqrt{\log n}} &\implies N^{2/3}\epsilon^2 = e^{\Omega(\sqrt{\log n})}. \label{eq:N-def2}
  \end{align}
  Let $M$ be the number of usable traces. Since our alignment rule
  ensures that the probability of being usable is at least
  $\frac{1}{2} p^{2K_0}$, it follows by a Chernoff bound and
  \eqref{eq:N-def1} that
  \[ \P_\bx(M < N^{2/3}) \le \exp\(- \frac{2N^{4/3}}{N} \) = \exp\(-2 e^{\Omega(\sqrt{\log n})} \) \le \frac{1}{n^3}. \]
  Define
  \[ \bavg_j = \frac{1}{M} \sum_{\text{$\wt{\bx}_i$ is usable}} \hat{b}_j(\wt{\bx}_i). \]
  By another Chernoff bound and \eqref{eq:N-def2},
  \begin{align}
    \P_\bx\(|\bavg_j - \overline{b}_j| > \epsilon/4 \) &\le \P_\bx(M < N^{2/3}) + \exp\( -\frac{-N^{2/3} \epsilon^2}{8} \) \nonumber \\
    & \le \frac{1}{n^3} + \exp\( - e^{\Omega(\sqrt{\log n})} \) \le \frac{1}{n^2}. \label{eq:b-avg}
  \end{align}
  Combining \eqref{eq:b-bar} and \eqref{eq:b-avg}, we conclude that
  \[ \P_\bx\( \text{$|\bavg_j - b_j(\bx^{(H+1):})| < \epsilon/2$ for all $j \le n$} \) \ge 1 - \frac{1}{n^2}. \]
  Thus, with probability at least $1 - \frac{1}{n^2}$, the conclusion
  of Lemma \ref{lemma:reconstruction-template} allows us to determine
  the first $K_3$ bits of $\bx^{(H+1):}$. As noted earlier, this
  includes the $(k+1)$-th bit of $\bx$, as desired.
\end{proof}

\subsection{Completing the proof}

We are finally ready to prove Theorem
\ref{thm:subpoly-reconstruction}, which is mostly a matter of
combining Lemmas \ref{lemma:initial-segment} and
\ref{lemma:advance-bit}.

\begin{proof}[Proof of Theorem \ref{thm:subpoly-reconstruction}]
  We sample $N = \ceil{\exp\(C'_p \sqrt{\log n}\)}$ traces, where
  $C'_p$ is large enough so that Lemmas
  \ref{lemma:initial-segment} and \ref{lemma:advance-bit} apply.

  We first condition on a realization $\bX = \bx$, and suppose that
  $\bx$ is good. We will construct a string $\hat{\bx} = (\hat{x}_1,
  \hat{x}_2, \ldots , \hat{x}_n)$. Let $E_k$ denote the event that
  $\hat{\bx}$ matches $\bx$ in the first $k$ bits. We construct the
  first $K_2$ bits of $\hat{\bx}$ using Lemma
  \ref{lemma:initial-segment}, which yields
  \begin{equation} \label{eq:initial-bound}
    \P_\bx(E_{K_2}) \ge 1 - \frac{1}{n}.
  \end{equation}

  Next, consider any $k$ with $K_2 \le k \le n/2$, and suppose we have
  constructed $\hat{x}_1, \ldots , \hat{x}_k$ already. We apply the
  algorithm of Lemma \ref{lemma:advance-bit} and set $\hat{x}_{k+1}$
  to its output. Although we do not have access to the first $k$ bits
  of $\bx$, we use the first $k$ bits of $\hat{\bx}$ instead. As long
  as $E_k$ holds, this will give us the correct value for $\hat{x}_{k
    + 1}$ with probability at least $1 - \frac{1}{n^2}$. Thus,
  \begin{equation} \label{eq:successive-bound}
    \P_\bx(E_{k+1}) \ge \P_\bx(E_k) - \frac{1}{n^2}.
  \end{equation}

  Using \eqref{eq:initial-bound} followed by repeated applications of
  \eqref{eq:successive-bound}, we find that
  \[ \P_\bx(E_{\ceil{n/2}}) \ge 1 - \frac{1}{n}. \]
  By symmetry, we can repeat the same procedure in reverse to
  reconstruct the last $\ceil{n/2}$ bits of $\bx$. Accounting for both
  the forward and reverse steps, the probability of failure is at most
  $\frac{2}{n}$.

  The final possible mode of failure is if $\bx$ is not good. However,
  by Lemma \ref{lemma:most-good}, this only happens with probability
  at most $\frac{1}{n}$. In total, we can reconstruct $\bX$ with
  probability at least $1 - \frac{3}{n}$. Moreover, we have only used
  $N = e^{O_p(\sqrt{\log n})}$ traces. This completes the proof.
\end{proof}


\bibliographystyle{plain}
\bibliography{rr}

\end{document}